\documentclass[conference]{IEEEtran}

\usepackage{graphicx,color}   %color for pdf_t generated by xfig
\usepackage[cmex10]{amsmath}  %cmex10 ensures that only type 1 fonts will be used
\usepackage{amssymb}
\usepackage{cite}  %combined consecutive references
\usepackage[utf8]{inputenc}
\usepackage[T1]{fontenc}

\usepackage{subfigure}
\usepackage{multirow}

\newtheorem{theorem}{Theorem}
\newtheorem{remark}{Remark}
\newtheorem{lemma}{Lemma}

\newcommand{\mbf}[1]{\ensuremath{\boldsymbol{#1}}}

\newcommand{\gap}{\vspace{-0.19ex}}

\begin{document}

%---------- Title ----------
\title{The Finite Field Multi-Way Relay Channel with Correlated Sources: The Three-User Case}
\author{\IEEEauthorblockN{Lawrence Ong$^\dag$, Roy Timo$^\ddag$, Gottfried Lechner$^\ddag$, Sarah J. Johnson$^\dag$, and Christopher M. Kellett$^\dag$}
\IEEEauthorblockA{$^\dag$School of Electrical Engineering and Computer Science, 
The University of Newcastle, Australia\\
$^\ddag$Institute for Telecommunications Research, University of South Australia, Australia\\
Email: lawrence.ong@cantab.net, \{roy.timo, gottfried.lechner\}@unisa.edu.au, \{sarah.johnson, chris.kellett\}@newcastle.edu.au}
}

\maketitle

\begin{abstract}
The three-user finite field multi-way relay channel with correlated sources is considered. The three users generate possibly correlated messages, and each user is to transmit its message to the two other users reliably in the Shannon sense. As there is no direct link among the users, communication is carried out via a relay, and the link from the users to the relay and those from the relay to the users are finite field adder channels with additive noise of arbitrary distribution. The problem is to determine the set of all possible achievable rates, defined as channel uses per source symbol for reliable communication. For two classes of source/channel combinations, the solution is obtained using Slepian-Wolf source coding combined with functional-decode-forward channel coding.
\end{abstract}

\section{Introduction}

In this paper we study the three-user finite field multi-way relay channel (MWRC) with correlated sources, where each user is to transmit its data to the other two users reliably (in the Shannon sense) via a single relay. For two classes of source/channel combinations, we obtain a complete characterization for reliable communication, i.e., the set of all achievable rates, defined as channel uses per source symbol.

% Correlated sources are commonly found in multiple geographically distributed measurements of the same type, e.g., temperature.
The MWRC is an extension of the two-way relay channel where two users exchange data via a relay~\cite{knopp06,rankovwittneben06,rankovwittneben07,kattigollakota07,gunduztuncel08,schnurrstanczak08}.
One application of the MWRC with correlated sources is the communication of weather stations via a satellite, where multiple stations obtain measurements of their respective local weather condition, and each station is to obtain the weather conditions at all other stations by communicating with a satellite.

We study the MWRC in which each user is to decode the data from all other users, and where there is no direct link among the users. Communication is carried out via a single relay. This system has been studied from the point of view of channel coding and source coding, independently, under different setups.

In the channel coding setups, the sources are assumed to be independent, and the channel noisy. The problem formulation is ``how many bits of data can each user send per channel use?'' \emph{Achievable rate tuples} here refer to the tuple of the number of message bits (per channel use) the users can transmit such that all other users can reliably recover their messages. The challenge is to find the \emph{capacity region} which is the closure of all achievable rate tuples. Though the capacity region of the general MWRC remains unknown, G\"und\"uz et al.~\cite{gunduzyener09} obtained asymptotic capacity results for the high SNR and the low SNR regimes for the Gaussian MWRC, and Ong et al.~\cite{ongjohnsonkellett10cl,ongmjohnsonit11} derived the capacity region of the finite field MWRC, which is achieved by \emph{functional-decode-forward} channel coding.

% For the Gaussian MWRC with independent sources, Gündüz \emph{et al.}~\cite{gunduzyener09} obtained asymptotic capacity results for the high SNR and the low SNR regimes. Ong \emph{et al.} \cite{ongjohnsonkellett10cl} proposed the functional-decode-forward (FDF) coding strategy for the binary MWRC with independent sources and showed that it achieved the common-rate capacity (where all users transmit at the same rate). FDF was extended to the finite field MWRC with independent sources and showed that it achieves the general-rate capacity \cite{ong10iccs} (where the users can transmit at different rates). The capacity of the general MWRC with independent sources, however, remains unknown to date.

In the source coding setups, the sources are assumed to be correlated, but the channel noiseless. The problem formulation is ``how many bits does each node need to encode per message symbol?'' \emph{Achievable rate tuples} here refer to the tuple of the number of bits (per source symbol) allocated to the encoders for which reliable reconstruction of each message is possible by all other users. The challenge is to find the set of all achievable rate tuples. The source coding problem for the three-user MWRC was solved by Wyner et al.~\cite{wynerwolf02}, using \emph{cascaded Slepian-Wolf} source coding~\cite{slepianwolf73}. 

% The two-user lossless case and lossy case (where each user reconstructs the other user's message with a prescribed distortion) was studied by Su and El Gamal~\cite{sugamal10}, and the two-user lossy case with common reconstructions (where each user must also be able to determine the lossy reconstructed message of the other user) was studied by Timo~\emph{et al.}~\cite{roysubmitted}.

In this paper, we study both source and channel coding in one setup, i.e., the three-user MWRC with noisy channel and with correlated sources (see our recent work~\cite{timoonglechner11} on the two-user MWRC with correlated sources). % We do not restrict ourselves to separate source-channel coding.
In the multi-terminal network, it is well known that solving the source coding and the channel coding problems separately does not solve the source-channel problem, i.e., noisy channels with correlated sources (see, e.g., the multiple-access channel \cite{dueck81}). We will, however, show that a strategy using Slepian-Wolf source coding (an optimal source coding) and functional-decode-forward channel coding (an optimal channel coding) is optimal for two classes of the finite field MWRCs with correlated sources. % This, incidentally, also proves source-channel coding separation in these classes of channels.
Using an example, we will also show that this scheme might not be optimal in general.

%Here, the users each send $m$ source symbols (to all other users) in $n$ channel uses. We say that the \emph{rate} of $\kappa = n/m$ channel uses per source symbol is achievable if all users can reliably decode the messages of all other users. Our problem is to find the set of all achievable rates, which is equivalent to finding the necessary and sufficient conditions for achievable rates. In this paper, we focus on the finite field channel for the three-user case (see \cite{raymedard03allerton,raymedard03globecom,nazergastpar07} for examples of finite field channels). On the one hand, for the three-user finite field MWRC with independent sources, the capacity is known and is achievable using FDF \cite{ong10iccs}. On the other hand, for the three-user noiseless MWRC with correlated sources, the set of all achievable rate tuples has been found and is achievable using cascaded Slepian-Wolf source coding \cite{wynerwolf02}. One is tempted to conclude that the combination of these two techniques solves the problem of the three-user finite field MWRC with correlated sources. In this paper, we show that the combination of these two techniques solve two classes of the channels, i.e., finding the necessary and sufficient conditions for achievable rates. Using an example, we show that this scheme, however, might not be optimal in general.

\section{Main Results}
\begin{figure}[t]
\centering
\resizebox{8.5cm}{!}{
\begin{picture}(0,0)%
\includegraphics{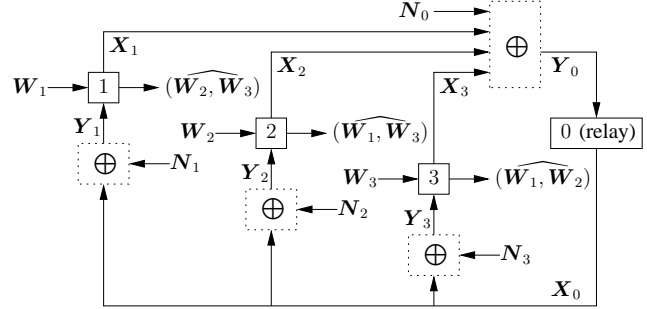}%
\end{picture}%
\setlength{\unitlength}{4144sp}%
\begingroup\makeatletter\ifx\SetFigFont\undefined%
\gdef\SetFigFont#1#2#3#4#5{%
  \fontsize{#1}{#2pt}%
  \fontfamily{#3}\fontseries{#4}\fontshape{#5}%
  \selectfont}%
\fi\endgroup%
\begin{picture}(5607,2736)(256,-2413)
\put(1621,-511){\makebox(0,0)[lb]{\smash{{\SetFigFont{12}{14.4}{\familydefault}{\mddefault}{\updefault}{\color[rgb]{0,0,0}$(\widehat{\mbf{W}_2,\mbf{W}_3})$}%
}}}}
\put(3106,-916){\makebox(0,0)[lb]{\smash{{\SetFigFont{12}{14.4}{\familydefault}{\mddefault}{\updefault}{\color[rgb]{0,0,0}$(\widehat{\mbf{W}_1,\mbf{W}_3})$}%
}}}}
\put(4546,-1321){\makebox(0,0)[lb]{\smash{{\SetFigFont{12}{14.4}{\familydefault}{\mddefault}{\updefault}{\color[rgb]{0,0,0}$(\widehat{\mbf{W}_1,\mbf{W}_2})$}%
}}}}
\put(1036,-511){\makebox(0,0)[lb]{\smash{{\SetFigFont{12}{14.4}{\familydefault}{\mddefault}{\updefault}{\color[rgb]{0,0,0}$1$}%
}}}}
\put(2521,-916){\makebox(0,0)[lb]{\smash{{\SetFigFont{12}{14.4}{\familydefault}{\mddefault}{\updefault}{\color[rgb]{0,0,0}$2$}%
}}}}
\put(4636,-151){\makebox(0,0)[lb]{\smash{{\SetFigFont{12}{14.4}{\familydefault}{\mddefault}{\updefault}{\color[rgb]{0,0,0}$\bigoplus$}%
}}}}
\put(3916,-1996){\makebox(0,0)[lb]{\smash{{\SetFigFont{12}{14.4}{\familydefault}{\mddefault}{\updefault}{\color[rgb]{0,0,0}$\bigoplus$}%
}}}}
\put(3961,-1321){\makebox(0,0)[lb]{\smash{{\SetFigFont{12}{14.4}{\familydefault}{\mddefault}{\updefault}{\color[rgb]{0,0,0}$3$}%
}}}}
\put(2476,-1591){\makebox(0,0)[lb]{\smash{{\SetFigFont{12}{14.4}{\familydefault}{\mddefault}{\updefault}{\color[rgb]{0,0,0}$\bigoplus$}%
}}}}
\put(1126,-151){\makebox(0,0)[lb]{\smash{{\SetFigFont{12}{14.4}{\familydefault}{\mddefault}{\updefault}{\color[rgb]{0,0,0}$\mbf{X}_1$}%
}}}}
\put(4051,-511){\makebox(0,0)[lb]{\smash{{\SetFigFont{12}{14.4}{\familydefault}{\mddefault}{\updefault}{\color[rgb]{0,0,0}$\mbf{X}_3$}%
}}}}
\put(811,-871){\makebox(0,0)[lb]{\smash{{\SetFigFont{12}{14.4}{\familydefault}{\mddefault}{\updefault}{\color[rgb]{0,0,0}$\mbf{Y}_1$}%
}}}}
\put(991,-1186){\makebox(0,0)[lb]{\smash{{\SetFigFont{12}{14.4}{\familydefault}{\mddefault}{\updefault}{\color[rgb]{0,0,0}$\bigoplus$}%
}}}}
\put(1666,-1186){\makebox(0,0)[lb]{\smash{{\SetFigFont{12}{14.4}{\familydefault}{\mddefault}{\updefault}{\color[rgb]{0,0,0}$\mbf{N}_1$}%
}}}}
\put(2296,-1276){\makebox(0,0)[lb]{\smash{{\SetFigFont{12}{14.4}{\familydefault}{\mddefault}{\updefault}{\color[rgb]{0,0,0}$\mbf{Y}_2$}%
}}}}
\put(3151,-1591){\makebox(0,0)[lb]{\smash{{\SetFigFont{12}{14.4}{\familydefault}{\mddefault}{\updefault}{\color[rgb]{0,0,0}$\mbf{N}_2$}%
}}}}
\put(3736,-1681){\makebox(0,0)[lb]{\smash{{\SetFigFont{12}{14.4}{\familydefault}{\mddefault}{\updefault}{\color[rgb]{0,0,0}$\mbf{Y}_3$}%
}}}}
\put(4591,-1996){\makebox(0,0)[lb]{\smash{{\SetFigFont{12}{14.4}{\familydefault}{\mddefault}{\updefault}{\color[rgb]{0,0,0}$\mbf{N}_3$}%
}}}}
\put(5041,-2311){\makebox(0,0)[lb]{\smash{{\SetFigFont{12}{14.4}{\familydefault}{\mddefault}{\updefault}{\color[rgb]{0,0,0}$\mbf{X}_0$}%
}}}}
\put(5041,-331){\makebox(0,0)[lb]{\smash{{\SetFigFont{12}{14.4}{\familydefault}{\mddefault}{\updefault}{\color[rgb]{0,0,0}$\mbf{Y}_0$}%
}}}}
\put(2611,-331){\makebox(0,0)[lb]{\smash{{\SetFigFont{12}{14.4}{\familydefault}{\mddefault}{\updefault}{\color[rgb]{0,0,0}$\mbf{X}_2$}%
}}}}
\put(271,-511){\makebox(0,0)[lb]{\smash{{\SetFigFont{12}{14.4}{\familydefault}{\mddefault}{\updefault}{\color[rgb]{0,0,0}$\mbf{W}_1$}%
}}}}
\put(1756,-916){\makebox(0,0)[lb]{\smash{{\SetFigFont{12}{14.4}{\familydefault}{\mddefault}{\updefault}{\color[rgb]{0,0,0}$\mbf{W}_2$}%
}}}}
\put(3196,-1321){\makebox(0,0)[lb]{\smash{{\SetFigFont{12}{14.4}{\familydefault}{\mddefault}{\updefault}{\color[rgb]{0,0,0}$\mbf{W}_3$}%
}}}}
\put(5131,-916){\makebox(0,0)[lb]{\smash{{\SetFigFont{12}{14.4}{\familydefault}{\mddefault}{\updefault}{\color[rgb]{0,0,0}$0$  (relay)}%
}}}}
\put(3691,164){\makebox(0,0)[lb]{\smash{{\SetFigFont{12}{14.4}{\familydefault}{\mddefault}{\updefault}{\color[rgb]{0,0,0}$\mbf{N}_0$}%
}}}}
\end{picture}%
}
\caption{The three-user finite field MWRC with correlated sources}
\label{fig:mwrc}
%\vspace{-2ex}
\end{figure}

\subsection{Channel Model}
Referring to Fig.~\ref{fig:mwrc}, nodes $1$, $2$, and $3$ are the users and node $0$ the relay. 
Let $W_i$ be the message of user $i$, and each \emph{message triplet} $(W_1,W_2,W_3)$ be independently generated according to $p(w_1,w_2,w_3)$. The memoryless finite field channel is defined by (i) the \emph{uplink:} $Y_0 = X_1 \oplus X_2 \oplus X_3 \oplus N_0$ and (ii) the \emph{downlinks:} $Y_i = X_0 \oplus N_i$, for $i \in \{1,2,3\}$, where $X_j,Y_j,N_j \in \mathcal{F}$, $\forall j \in \{0,1,2,3\}$, where $\mathcal{F}$ is a finite field with addition $\oplus$, $X_j$ is the channel input from node $j$, $Y_j$ is the channel output received by node $j$, and $N_j$ is the receiver noise of node $j$. The noise terms $N_j$'s are independent for each user and for each channel use. If the users can reliably\footnote{Reliability is defined in the Shannon sense, i.e., the probability that any user wrongly decodes any other user's message can be made as small as desired.} exchange $m$ message triplets in $n$ channel uses, the rate of $\kappa=n/m$ channel uses per source symbol is said to be \emph{achievable}.

\subsection{Definitions} \label{sec:definition}
Before presenting the main results of this paper, we define the following properties for the sources and the channel:\\
The channel is
\begin{itemize}
\item \emph{symmetrical}, if all the downlinks from the relay to the users are equally noisy, i.e., $H(N_1) = H(N_2) = H(N_3)$;
\item \emph{asymmetrical}, otherwise.
\end{itemize}
%Here, $H(X) = \sum_{x} p(x) \log_2\frac{1}{p(x)}$ is the entropy of the random variable $X$.

\begin{remark}
For a symmetrical channel, we do not impose the constraint that the ``noise power'' at the relay, $H(N_0)$, must equal that at the users. This means that all downlinks from the relay to the users are equally noisy, but the uplink from the users to the relay can be noisier or less noisy. %An asymmetric channel is one which is not symmetric, i.e., different channels may be subject to different levels of noise.
\end{remark}

The sources have
\begin{itemize}
\item \emph{almost-balanced conditional mutual information} (ABCMI), if
\gap
\begin{multline}
I(W_i;W_j|W_k) \leq I(W_j;W_k|W_i)+ I(W_i;W_k|W_j), \\ \forall i,j,k \in \{1,2,3\} \text{ and } i \neq j \neq k; \label{eq:balanced}
\end{multline}
\item \emph{unbalanced conditional mutual information}, otherwise, i.e., there exists some user $A \in \{1,2,3\}$, such that
\gap
\begin{multline}
I(W_B;W_C|W_A) = I(W_A;W_B|W_C)\\ + I(W_A;W_C|W_B) + \eta, \label{eq:unbalanced}
\end{multline}
for some $\eta > 0$ and $B,C \in \{1,2,3\}\setminus \{A\}$ where $B \neq C$.
\end{itemize}
%Here, $I(X;Y|Z) = \sum_{x,y,z} p(x,y,z) \log_2 \frac{p(x,y|z)}{p(x|z)p(y|z)}$ is the conditional mutual information of random variables $X$ and $Y$ given $Z$.
Note that if \eqref{eq:balanced} fails, then there necessarily exists an appropriate $\eta$ as in \eqref{eq:unbalanced}.

For the class of sources with unbalanced conditional mutual information, we say that the sources have \emph{skewed conditional entropies} (SCE) if in addition to \eqref{eq:unbalanced}, we have the following:
\gap
\begin{multline}
H(W_B,W_C|W_A)  \geq \max \Big\{ H(W_A,W_B|W_C) ,\\ H(W_A,W_C|W_B) \Big\} + \eta, \label{eq:case2a}
\end{multline}
for the same $\eta$ as in \eqref{eq:unbalanced}.
%Here, $H(X,Y|Z) = \sum_{x,y,z} p(x,y,z) \log_2 \frac{1}{p(x,y|z)}$ is the conditional entropy.

Fig.~\ref{fig:case1-2} shows the relationship among the entropies and mutual information for the three source messages $W_1$, $W_2$, and $W_3$ for the cases of ABCMI and SCE. Referring to Fig.~\ref{fig:case-1}, the shaded areas refer to the mutual information between any two source messages given the third source message. For ABCMI, we have that any of the three shaded areas must not be bigger than the sum of the other two shaded areas. Suppose that the sources do not have ABCMI, then they must have unbalanced conditional mutual information, i.e., we can find a user $A$ where $I(W_B;W_C|W_A)$ is larger than the sum of $I(W_A;W_B|W_C)$ and $I(W_A;W_C|W_B)$ by an amount $\eta$ (see Fig.~\ref{fig:case-2a}). In addition, for sources with SCE, we also have that for the two messages, $W_B$ and $W_C$, whose mutual information conditioned on $W_A$, i.e., $I(W_B;W_C|W_A)$, is larger than the sum of the other two pairs by the amount $\eta$, their entropy conditioned on $W_A$, i.e., $H(W_B,W_C|W_A)$, is also greater than that of any other pair (conditioned on the message of the third user) by at least $\eta$. %The entropy diagram for SCE is depicted in Fig.~\ref{fig:case-2a}.

\begin{figure}[t]
\centering
\subfigure[almost-balanced conditional mutual information (ABCMI)]{
\includegraphics[width=0.78\linewidth]{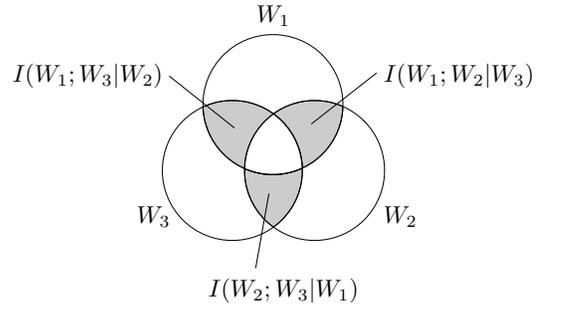}
\label{fig:case-1}
}
\subfigure[skewed conditional entropies (SCE)]{
\includegraphics[width=0.85\linewidth]{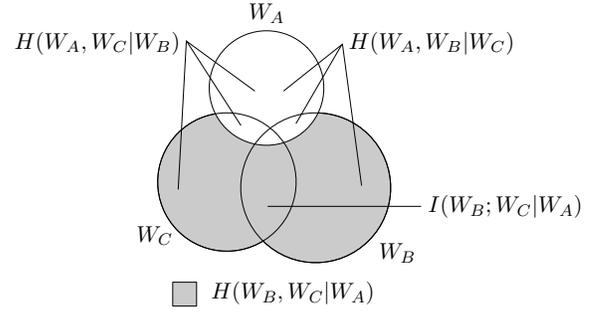}
\label{fig:case-2a}
}
\caption{Information diagrams for the sources}
\label{fig:case1-2}
%\vspace{-2ex}
\end{figure}

\subsection{Main Results}
In this paper, we show that for the two cases of (i) sources with ABCMI, and (ii) sources with SCE on symmetrical channels, necessary and sufficient conditions for a rate $\kappa>0$ to be achievable are
\gap
\begin{multline}
H(W_j,W_k|W_i) \leq \kappa \Big[\log_2 |\mathcal{F}| - \max \{ H(N_0), H(N_i) \} \Big],\\ \forall i,j,k \in \{1,2,3\} \text{ and } i \neq j \neq k. \label{eq:main}
\end{multline}
For the remaining cases, % i.e., (i) sources with SCE on asymmetrical channels, and (ii) sources without SCE but with unbalanced conditional mutual information,
the above conditions are necessary, but may not be sufficient. For these cases, we obtain sufficient conditions for achievability, but there remains a gap between \eqref{eq:main} and the sufficient conditions presented. The results are summarized in Table~\ref{table:main-result}.

\begin{table*}[t]
\centering
\caption{Main Results: The Set of All Achievable Rates Found for the Following Cases}
\label{table:main-result}
\begin{tabular}{|c | c | c | c |}
\hline
\multirow{2}{*}{Source structure} & Almost-balanced conditional & \multicolumn{2}{c |}{ Unbalanced conditional mutual information} \\
\cline{3-4}
& mutual information (ABCMI) & Skewed conditional entropies (SCE) & Others \\
\hline 
\hline
Symmetrical channel & $\surd$ & $\surd$ & $\times$ \\
Asymmetrical channel & $\surd$ & $\times$ & $\times$\\
\hline
\end{tabular}
\end{table*}

\begin{remark}
Utilizing \emph{feedback} is permitted in our system model. Consider $m$ message triplets, and define $\mbf{W}_i \triangleq (W_i[1], \dotsc, W_i[m])$.  
The transmitted channel symbol of node $i$ at any time $t$ is a function of its own messages and its past received channel symbols, i.e, $X_i[t] = f_{i,t} (\mbf{W}_i, Y_i[1], \dotsc, Y_i[t])$, $\forall t \in \{1,\dotsc,n\}, \forall i \in \{0,1,2,3\}$, where $\mbf{W}_0 \triangleq \varnothing$. After $n$ channel uses, each user $i$ then estimates the messages of the other users using its own messages and its $n$ received channel symbols, i.e., $(\widehat{\mbf{W}_j,\mbf{W}_k}) = g_{i} (\mbf{W}_i,Y_i[1], \dotsc, Y_i[n])$, $\forall i,j,k \in \{1,2,3\}$ where $i \neq j \neq k$, and $j < k$.
\end{remark}

Now, we will prove the above results in Sections~\ref{section:conditions}--\ref{section:necessary-and-sufficient}.

\section{Conditions for Achievability} \label{section:conditions}

\subsection{Necessary Conditions for Achievability}

We can show the following necessary conditions for achievability:
\begin{theorem}\label{theorem:necessary-condition}
Consider a three-user finite field MWRC with correlated sources. If a rate $\kappa>0$ is achievable then 
\gap
\begin{equation}
%H(W_i|W_j,W_k) \leq\log_2|\mathcal{F}| - H(N_0)\\
H(W_j,W_k|W_i)  \leq \kappa \Big[\log_2|\mathcal{F}| - \max \{ H(N_0), H(N_i) \} \Big], \label{eq:necessary-condition}
\end{equation}
for all $i,j,k \in \{1,2,3\}$ where $i \neq j \neq k$.
\end{theorem}

To prove the above necessary conditions for achievability, we first extend the cut-set argument in \cite[pages 587--591]{coverthomas06} for networks with independent sources to those with correlated sources. We then apply the cut-set argument for networks with correlated sources to the three-user finite field MWRC. The proof is omitted because of space constraints. %{\footnotesize (Proof available for ISIT reviewing in additional Appendix A)} % We present the proof in Appendix~\ref{appendix:necessary}.

\subsection{Sufficient Conditions for Achievability}

We can also show the following sufficient condition for achievability:

\begin{theorem}\label{theorem:sufficient-condition}
Consider a three-user finite field MWRC with correlated sources. For any $\kappa >0$, % The rate $\kappa$ is achievable if there exists three non-negative real numbers $R_1$, $R_2$, and $R_3$ such that
if  there exist three non-negative real numbers $R_1$, $R_2$, and $R_3$ satisfying
\gap
\begin{align}
% R_1 &\geq\frac{1}{\kappa} H(W_1|W_2,W_3) \label{eq:sufficient-rate-a}\\
% R_2 &\geq \frac{1}{\kappa} H(W_2|W_1,W_3) \label{eq:sufficient-rate-b} \\
% R_3 &\geq \frac{1}{\kappa} H(W_3|W_1,W_2) \label{eq:sufficient-rate-c} \\
% R_1+R_2 &\geq \frac{1}{\kappa} H(W_1,W_2|W_3) \label{eq:sufficient-rate-d} \\
% R_1+R_3 &\geq \frac{1}{\kappa} H(W_1,W_3|W_2) \label{eq:sufficient-rate-e} \\
% R_2+R_3 &\geq \frac{1}{\kappa} H(W_2,W_3|W_1) \label{eq:sufficient-rate-f} \\
% R_1+R_2 &\leq \log_2|\mathcal{F}| - \max \{ H(N_0), H(N_3) \} \label{eq:sufficient-rate-g} \\
% R_2+R_3 &\leq \log_2|\mathcal{F}| - \max \{ H(N_0), H(N_2) \} \label{eq:sufficient-rate-h} \\
% R_2+R_3 &\leq \log_2|\mathcal{F}| - \max \{ H(N_0), H(N_1) \}. \label{eq:sufficient-rate-i}
R_i &\geq\frac{1}{\kappa} H(W_i|W_j,W_k) \label{eq:sufficient-rate-a}\\
R_j+R_k &\geq \frac{1}{\kappa} H(W_j,W_k|W_i) \label{eq:sufficient-rate-d} \\
R_j+R_k &\leq \log_2|\mathcal{F}| - \max \{ H(N_0), H(N_i) \}, \label{eq:sufficient-rate-g}
\end{align}
for all $i,j,k \in \{1,2,3\}$ where $i \neq j \neq k$, then the rate $\kappa$ is achievable.
\end{theorem}

\begin{remark}
$R_i$ in the above theorem is the ``rate'' (in bits per channel use) of the encoded message at which node $i$ transmits on the uplink, for $i \in \{1,2,3\}$.
\end{remark}

{\it Sketch of proof for Theorem~\ref{theorem:sufficient-condition}:} To prove the above sufficient condition for achievability, we use the techniques of (i) Slepian-Wolf source coding for the noiseless MWRC with correlated sources \cite{wynerwolf02}, and (ii) functional-decode-forward channel coding for the finite field MWRC with independent sources \cite{ongmjohnsonit11}. We first use Slepian-Wolf source coding: each user $i$ encodes its message $\mbf{W}_i$ to an $nR_i$-bit message $M_i$, for $i \in \{1,2,3\}$. We then use functional-decode-forward channel coding on the finite field (noisy) channel: each user $i$ transmits $M_i$ using a random linear code; the relay decodes a linear function of messages $(M_1,M_2,M_3)$, denoted by $\mbf{U}$, and forwards $\mbf{U}$ back to the users. If \eqref{eq:sufficient-rate-g} is satisfied, each user $i$ can decode the encoded messages of the other users, i.e., $M_j$ and $M_k$. Furthermore, if \eqref{eq:sufficient-rate-a}--\eqref{eq:sufficient-rate-d} are satisfied, each user $i$ can recover $(\mbf{W}_j,\mbf{W}_k)$ from $(\mbf{W}_i,M_j,M_k)$. Note that the source codes and the channel codes can be designed separately and independently. The detail of the proof is omitted because of space constraints.  %{\footnotesize (Proof available for ISIT reviewing in additional Appendix B)}  \hfill $\blacksquare$
 % We present the proof in Appendix~\ref{appendix:sufficient}.

In the rest of this paper, we denote the above coding strategy of using Slepian-Wolf source coding and functional-decode-forward channel coding by SW-FDF.

\begin{remark}
SW-FDF proposed here is not simply selecting appropriate channels to support the source coding operation in \cite{wynerwolf02}. In \cite{wynerwolf02}, the relay obtains $M_i$ from each user $i$ on the uplink, $\forall i$. It then creates \emph{random bins} for $(M_1,M_2,M_3)$, and sends the bin index $M_0$ to the users. % Actually, having $(M_1,M_2,M_3)$, the relay has the flexibility to form any function of them.
In SW-FDF, we do not merely find channels to support the above, i.e., uplinks that allow the relay to decode $(M_1,M_2,M_3)$, and downlinks that allow the relay to send $M_0$ to the users. In fact, we previously showed that having the relay fully decode $(M_1,M_2,M_3)$ can be sub-optimal~\cite{ongmjohnsonit11}. In SW-FDF, we first use the idea in Slepian-Wolf source coding to generate $(M_1,M_2,M_3)$, and then use functional-decode-forward for the relay to \emph{decode a function} of $(M_1,M_2,M_3)$, which depends on the channel codes chosen, and then forward this function to the users.
\end{remark}

 %Hence, we do not have constraints for an independent ``downlink'' variable $M_0$ or ``downlink'' rate $R_0$ here (c.f. \cite[the last constraint in (12)]{wynerwolf02}). However, the downlink channels are constrained by the  On the other hand, the  we are subject to the constraints on channel coding for the noisy channel instead, c.f., \eqref{eq:sufficient-rate-g}.

\begin{remark}
Conditions \eqref{eq:sufficient-rate-d}--\eqref{eq:sufficient-rate-g} imply \eqref{eq:necessary-condition} in Theorem~\ref{theorem:necessary-condition}. In addition, we require \eqref{eq:sufficient-rate-a} for all $i, j, k \in \{ 1, 2, 3\}$ where $i \neq j \neq k$, which are source coding constraints, to guarantee reliable communication using SW-FDF.
\end{remark}

\section{Necessary and Sufficient Conditions} \label{section:necessary-and-sufficient}

Now, we find conditions under which the above necessary and sufficient conditions match. Note that any three-user finite field MWRC must belong to one of the following two classes.
%We will consider the following two cases, of which one of them must be true:

\subsection{Sources with Almost-Balanced Conditional Mutual Information (ABCMI)}
For ABCMI, we have the following constraints:
\gap
\begin{multline}
I(W_i;W_j|W_k) \leq I(W_j;W_k|W_i)+ I(W_i;W_k|W_j), \\ \forall i,j,k \in \{1,2,3\} \text{ and } i \neq j \neq k. \label{eq:case-1}
\end{multline}

We first show the following lemma.
\begin{lemma}\label{lemma:case-1}
Consider three (possibly correlated) random variables $W_1$, $W_2$, $W_3$, and a positive number $\kappa$. If \eqref{eq:case-1} is true, then we can always find three non-negative real numbers $R_1, R_2, R_3 \geq 0$, such that 
\gap
\begin{align}
\kappa R_1 &\geq H(W_1|W_2,W_3) \label{eq:lemma-case-1-a}\\
\kappa R_2 &\geq H(W_2|W_1,W_3) \label{eq:lemma-case-1-b} \\
\kappa  R_3 &\geq H(W_3|W_1,W_2) \label{eq:lemma-case-1-c} \\
\kappa (R_1+R_2) &= H(W_1,W_2|W_3) \label{eq:lemma-case-1-d} \\
\kappa  (R_1+R_3) &= H(W_1,W_3|W_2) \label{eq:lemma-case-1-e} \\
\kappa  (R_2+R_3) &= H(W_2,W_3|W_1). \label{eq:lemma-case-1-f}
\end{align}
\end{lemma}

\begin{proof}[Proof of Lemma~\ref{lemma:case-1}]
It can be shown that choosing
\gap
\begin{multline}
\kappa  R_i = H(W_i|W_j,W_k) + \frac{1}{2} \Big[ I(W_i;W_k|W_j)\\ + I(W_i;W_j|W_k) - I(W_j;W_k|W_i) \Big], \label{eq:case-1-r1-r3}
\end{multline}
for $i,j,k \in \{1,2,3\}$ and $i \neq j \neq k$ satisfies \eqref{eq:lemma-case-1-a}--\eqref{eq:lemma-case-1-f}.
\end{proof}

We now show that for ABCMI, i.e., when \eqref{eq:case-1} is true, the necessary conditions in Theorem~\ref{theorem:necessary-condition} are sufficient for reliable communications. Let $\kappa$ be any rate that satisfies \eqref{eq:necessary-condition}. First, choosing $R_1$, $R_2$, and $R_3$ in \eqref{eq:case-1-r1-r3}, from Lemma~\ref{lemma:case-1}, \eqref{eq:sufficient-rate-a}--\eqref{eq:sufficient-rate-d} are satisfied.
Substituting \eqref{eq:lemma-case-1-d}--\eqref{eq:lemma-case-1-f} into \eqref{eq:necessary-condition}, we get
\gap
\begin{equation}
R_j + R_k \leq \log_2 |\mathcal{F}| - \max \{ H(N_0), H(N_i) \},
\end{equation}
for all $i,j,k \in \{1,2,3\}$ where $i \neq j \neq k$. This is equivalent to \eqref{eq:sufficient-rate-g}. This means the rate $\kappa$ is achievable.

So, we have the following necessary and sufficient conditions for reliable communications for ABCMI:
\begin{theorem} [ABCMI]
Consider a three-user finite field MWRC with correlated sources. If the sources have ABCMI, then the rate $\kappa$ is achievable if and only if
\gap
\begin{equation}
H(W_j,W_k|W_i) \leq \kappa \Big[\log_2 |\mathcal{F}| - \max \{ H(N_0), H(N_i) \} \Big], \nonumber
\end{equation}
for all $i,j,k \in \{1,2,3\}$ where $i \neq j \neq k$.
\end{theorem}

% Since symmetrical sources (see Definition~\ref{definition:symmetrical-sources}) implies \eqref{eq:sym}, we have
% \begin{corollary}
% Consider a three-user finite field MWRC with correlated and symmetrical sources.
% Reliable communication is possible if and only if
% \begin{equation}
% H(W_j,W_k|W_i) \leq \min \{ C_0, C_i \},
% \end{equation}
% for all $i,j,k \in \{1,2,3\}$.
% \end{corollary}

\begin{remark}
If the sources are independent, we have $H(W_i,W_j|W_k) = H(W_i) + H(W_j)$ for all $i,j,k \in \{1,2,3\}$, $i \neq j \neq k$. In this case, \eqref{eq:case-1} is always satisfied. So, reliable communication is possible if and only if $r_j + r_k \leq \log_2 |\mathcal{F}| - \max \{ H(N_0), H(N_i) \}$, for all $i,j,k \in \{1,2,3\}$ and $i \neq j \neq k$, where $r_j = H(W_j)/\kappa = mH(W_j)/n$ is the number of message bits transmitted by user $j$ per channel use. With this we recover the capacity region of the three-user MWRC with independent sources~\cite{ongmjohnsonit11}.
\end{remark}

\subsection{Sources with Unbalanced Conditional Mutual Information}
For sources with unbalanced conditional mutual information, we have the following constraint: there exists some user $A \in \{1,2,3\}$, such that
\gap
\begin{equation}
I(W_B;W_C|W_A) = I(W_A;W_B|W_C) + I(W_A;W_C|W_B) + \eta, \label{eq:case-2}
\end{equation}
for some $\eta > 0$, for some $B,C \in \{1,2,3\} \setminus \{A\}$, and $ B \neq C$. We can show that if the sources do not have ABCMI (i.e., \eqref{eq:case-1} is false), they must have unbalanced conditional mutual information. Now we prove the following lemma:

% It follows that
% \begin{align}
% I(W_A;W_B|W_C) \leq I(W_B;W_C|W_A) + I(W_A;W_C|W_B) \label{eq:case-2-a}\\
% I(W_A;W_C|W_B) \leq I(W_B;W_C|W_A) + I(W_A;W_B|W_C). \label{eq:case-2-b}
% \end{align}
% This means the sources do not have ABCMI (i.e., \eqref{eq:case-1} is false), they must have unbalanced conditional mutual information (i.e., \eqref{eq:case-2}--\eqref{eq:case-2-b} must be true for some $A$).

\begin{lemma}\label{lemma:case-2}
Consider three (possibly correlated) random variables $W_1$, $W_2$, $W_3$, and a positive number $\kappa$. If \eqref{eq:case-2} is true, then we can always find three non-negative real numbers $R_A, R_B, R_C \geq 0$,  such that 
\gap
\begin{align}
\kappa R_A &= H(W_A|W_B,W_C) \label{eq:lemma-case-2-a}\\
\kappa R_B &> H(W_B|W_A,W_C) \label{eq:lemma-case-2-b} \\
\kappa R_C &> H(W_C|W_A,W_B) \label{eq:lemma-case-2-c} \\
\kappa (R_A+R_B) &= H(W_A,W_B|W_C) \label{eq:lemma-case-2-d} \\
\kappa (R_A+R_C) &= H(W_A,W_C|W_B) \label{eq:lemma-case-2-e} \\
\kappa (R_B+R_C) &= H(W_B,W_C|W_A), \label{eq:lemma-case-2-f}
\end{align}
for the $\eta>0$ in \eqref{eq:case-2}.
\end{lemma}

\begin{proof}[Proof of Lemma~\ref{lemma:case-2}]
Constraint \eqref{eq:case-2} implies that
\gap
\begin{align}
&I(W_B;W_C|W_A) + I(W_A;W_B|W_C) - I(W_A;W_C|W_B) \nonumber\\
&\quad \quad = 2 I(W_A;W_B|W_C) + \eta > 0 \label{eq:eta-1} \\
&I(W_B;W_C|W_A) + I(W_A;W_C|W_B) - I(W_A;W_B|W_C) \nonumber \\ 
&\quad \quad = 2 I(W_A;W_C|W_B) + \eta > 0. \label{eq:eta-2}
\end{align}

Now, we choose
\gap
\begin{align}
\kappa R_A &= H(W_A|W_B,W_C) \label{eq:case-2-r1}\\
\kappa R_B &= H(W_B|W_A,W_C) + \frac{1}{2} \Big[ I(W_B;W_C|W_A)\nonumber\\
&\quad + I(W_A;W_B|W_C) - I(W_A;W_C|W_B) \Big] \label{eq:case-2-r2}\\
\kappa R_C &= H(W_C|W_A,W_B) + \frac{1}{2} \Big[ I(W_B;W_C|W_A) \nonumber\\
&\quad + I(W_A;W_C|W_B) - I(W_A;W_B|W_C) \Big]. \label{eq:case-2-r3}
\end{align}

Substituting \eqref{eq:eta-1} into \eqref{eq:case-2-r2} and \eqref{eq:eta-2} into \eqref{eq:case-2-r3}, we have $\kappa R_B > H(W_B|W_A,W_C)$ and $\kappa R_C > H(W_C|W_A,W_B)$ respectively, i.e., \eqref{eq:lemma-case-2-b}--\eqref{eq:lemma-case-2-c} are satisfied.
Summing different pairs of \eqref{eq:case-2-r1}--\eqref{eq:case-2-r3}, we have $\kappa (R_A + R_B) = H(W_A,W_B|W_C) + \eta$, $\kappa (R_A + R_C) = H(W_A,W_C|W_B) + \eta$, and $\kappa (R_B + R_C) = H(W_B,W_C|W_A)$, i.e., \eqref{eq:lemma-case-2-d}--\eqref{eq:lemma-case-2-f} are satisfied.
\end{proof}

\subsubsection{Skewed Conditional Entropies (SCE)}

Now, we consider sources with skewed conditional entropies (SCE), i.e., those with unbalanced conditional mutual information with the following additional constraints:
\gap
\begin{align}
H(W_B,W_C|W_A) & \geq H(W_A,W_C|W_B) + \eta \label{eq:case-2a-1}\\
H(W_B,W_C|W_A) & \geq H(W_A,W_B|W_C) + \eta, \label{eq:case-2a-2}
\end{align}
which is equivalent to $H(W_B,W_C|W_A) \geq \max \{H(W_A,W_C|W_B), H(W_A,W_B|W_C)  \} + \eta$. 

For SCE, we have the following result:
\begin{theorem}[SCE] \label{theorem:2a}
Consider a three-user finite field MWRC with correlated sources. For SCE, i.e.,
\gap
\begin{align}
I(W_B;W_C|W_A) &= I(W_A;W_B|W_C) + I(W_A;W_C|W_B) + \eta, \label{eq:2a-extra-conditions-2}\\
H(W_B,W_C|W_A) & \geq \max \Big\{ H(W_A,W_C|W_B), \nonumber\\
&\quad\quad\quad\quad  H(W_A,W_B|W_C) \Big\}+ \eta, \label{eq:2a-extra-conditions}
\end{align}
for some $A,B,C \in \{1,2,3\}$, $A \neq B \neq C$, and some $\eta >0$,
if the channel is symmetrical (refer to the definition in Sec.~\ref{sec:definition}), then the rate $\kappa$ is achievable if and only if
\gap
\begin{equation}
H(W_j,W_k|W_i) \leq \kappa \Big[ \log_2 |\mathcal{F}| - \max \{ H(N_0), H(N_i) \} \Big], \label{eq:result-2a}
\end{equation}
for all $i,j,k \in \{1,2,3\}$ where $i \neq j \neq k$.
\end{theorem}

\begin{proof}[Proof of Theorem~\ref{theorem:2a}]
From Theorem~\ref{theorem:necessary-condition}, the rate $\kappa$ is achievable only if \eqref{eq:result-2a} is satisfied.

For the symmetrical channel, we can define $H(N_\text{d}) \triangleq H(N_1) = H(N_2) = H(N_3)$, where the subscript ``d'' denotes downlink. Eqn. \eqref{eq:2a-extra-conditions} for SCE implies that \eqref{eq:result-2a} reduces to
\gap
\begin{equation}
H(W_B,W_C|W_A) \leq \kappa \Big[ \log_2|\mathcal{F}| - \max\{ H(N_0), H(N_\text{d}) \} \Big]. \label{eq:necessary-2a}
\end{equation}

We need to show that for all symmetrical MWRCs where the sources have SCE, if \eqref{eq:necessary-2a} is satisfied then the rate $\kappa$ is achievable.

From Lemma~\ref{lemma:case-2}, we can choose three non-negative real numbers $R_A$, $R_B$, and $R_C$ such that conditions \eqref{eq:sufficient-rate-a}--\eqref{eq:sufficient-rate-d} in Theorem~\ref{theorem:sufficient-condition} are satisfied. In addition, for the chosen $R_A$, $R_B$, and $R_C$, we have
\gap
\begin{align}
\kappa (R_B + R_C) & = H(W_B,W_C|W_A) \nonumber\\
& \leq \kappa \Big[\log_2|\mathcal{F}| - \max\{ H(N_0), H(N_\text{d}) \} \Big] \label{eq:2a-1}\\
\kappa (R_A + R_B) & = H(W_A,W_B|W_C) + \eta \nonumber \\ &  \leq H(W_B,W_C|W_A) \nonumber\\
&\leq  \kappa \Big[\log_2|\mathcal{F}| - \max\{ H(N_0), H(N_\text{d}) \} \Big] \label{eq:2a-2}\\
\kappa (R_A + R_C) & = H(W_A,W_C|W_B) + \eta \nonumber \\ & \leq H(W_B,W_C|W_A) \nonumber\\
& \leq  \kappa \Big[\log_2|\mathcal{F}| - \max\{ H(N_0), H(N_\text{d}) \} \Big], \label{eq:2a-3}
\end{align}
where \eqref{eq:2a-1} follows from \eqref{eq:lemma-case-2-f} and \eqref{eq:necessary-2a}; \eqref{eq:2a-2} follows from \eqref{eq:lemma-case-2-d}, \eqref{eq:2a-extra-conditions}, and \eqref{eq:necessary-2a}; and \eqref{eq:2a-3} follows from \eqref{eq:lemma-case-2-e}, \eqref{eq:2a-extra-conditions}, and \eqref{eq:necessary-2a}. The above three conditions are equivalent to \eqref{eq:sufficient-rate-g} in Theorem~\ref{theorem:sufficient-condition}. This means the rate $\kappa$ is achievable.

So, for any symmetrical finite field MWRC with correlated sources with SCE, the rate $\kappa$ is achievable if and only if \eqref{eq:result-2a} is satisfied.
\end{proof}

% \begin{remark}
% We have shown that CSW-FDF is optimal for
% \begin{itemize}
% \item sources with ABCMI
% \item sources with CSE on a symmetrical channel.
% \end{itemize}
% For the above cases, we have shown that for any $\kappa$ that satisfies \eqref{eq:necessary-condition} in Theorem~\ref{theorem:necessary-condition}, we are able to find a ``rate triplet'' $(R_1,R_2,R_3)$ that satisfies \eqref{eq:sufficient-rate-a}--\eqref{eq:sufficient-rate-g} in Theorem~\ref{theorem:sufficient-condition}. More specifically, for ABCMI (which we do not require the MWRC to be symmetrical), $(R_1,R_2,R_3)$ can be chosen such that
% \begin{align}
% \kappa (R_1 + R_2)  &= H(W_1,W_2|W_3) \\
% \kappa (R_1 + R_3)  &= H(W_1,W_3|W_2) \\
% \kappa (R_2 + R_3)  &= H(W_2,W_3|W_1).
% \end{align}
% This ensures that any $\kappa$ that satisfies \eqref{eq:necessary-condition} will also satisfies \eqref{eq:sufficient-rate-d}--\eqref{eq:sufficient-rate-g}.
% \end{remark}

\subsubsection{Without Skewed Conditional Entropies}
Finally, consider sources without SCE but with unbalanced conditional mutual information, i.e., \eqref{eq:case-2} with the following additional constraint:
\gap
\begin{multline}
H(W_B,W_C|W_A) < \max \Big\{ H(W_A,W_C|W_B),\\ H(W_A,W_B|W_C) \Big\} + \eta.
\end{multline}

For this case, we will show that SW-FDF might not be optimal, i.e., not all rates $\kappa$ that satisfy Theorem~\ref{theorem:necessary-condition} are achievable using the coding scheme. We consider the source structure depicted in Fig.~\ref{fig:case-2b}, i.e., $H(W_A,W_B|W_C)=H(W_A,W_C|W_B) =7$ and $H(W_B,W_C|W_A) =5$, and the channel with the following parameters:  $\log_2|\mathcal{F}| - H(N_0) = 7$, $\log_2|\mathcal{F}| - H(N_i) = 10$, $\forall i \in \{1,2,3\}$. %Note that the channel is symmetrical.

\begin{figure}[t]
\centering
%\hspace{-15ex}
\includegraphics[width=3.24cm]{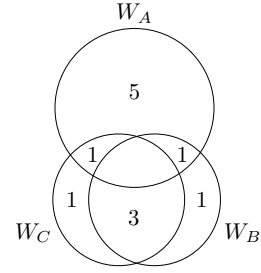}
\caption{An example of sources with unbalanced conditional mutual information and without SCE}
\label{fig:case-2b}
\vspace{-2ex}
\end{figure}

The rate of $\kappa =1$ is necessary for reliable communication (from Theorem~\ref{theorem:necessary-condition}). If $\kappa=1$ is achievable using SW-FDF, from Theorem~\ref{theorem:sufficient-condition}, we must be able to find three non-negative real numbers $R_A$, $R_B$, and $R_C$ such that
\gap
\begin{align}
R_A &\geq 5 \label{eq:contradiction-1}\\
\min\{R_B, R_C\} &\geq 1 \\
\min\{R_A+R_B, R_A + R_C\} &\geq 7 \\
R_B+R_C & \geq 5 \label{eq:contradiction-2} \\
\max\{R_A + R_B, R_A + R_C, R_B + R_C\} &\leq 7. \label{eq:contradiction-3}
\end{align}
From \eqref{eq:contradiction-1} and \eqref{eq:contradiction-2}, we must have $\max \{ R_A + R_B, R_A + R_C \} \geq 7.5$. This means \eqref{eq:contradiction-3} cannot be satisfied, and hence the rate of $\kappa=1$ is not achievable using SW-FDF.

% \section{Conclusion} \label{sec:conclusion}

% In this paper we have obtained the entire achievable rate region for two classes of the three-user finite field multi-way relay channel (MWRC) with correlated sources. The rates are achievable by cascaded Slepian-Wolf source coding and functional-decode-forward channel coding (CSW-FDF). We have selected the three-user case and the finite field channel to demonstrate that the problem of MWRC with correlated sources is not a straight-forward extension of the MWRC with independent sources and the noiseless MWRC with correlated sources. Even though FDF channel coding achieves the set of all achievable rate tuples for the MWRC with independent sources, and cascaded Slepian-Wolf source coding achieves that for the noiseless MWRC with correlated sources, the combination of these two techniques, CSW-FDF, fails to achieve the cut-set bound for the combination of these two channels, i.e., MWRC with correlated sources, in general. However, despite not being able to achieve the cut-set bound of all three-user finite field MWRCs with correlated sources, FDF still proves to be optimal for two classes of source/channel combinations.

% {\ttfamily To insert comments on generalisation to $N>3$.

\bibliography{arxiv}

\begin{thebibliography}{10}
\providecommand{\url}[1]{#1}
\csname url@samestyle\endcsname
\providecommand{\newblock}{\relax}
\providecommand{\bibinfo}[2]{#2}
\providecommand{\BIBentrySTDinterwordspacing}{\spaceskip=0pt\relax}
\providecommand{\BIBentryALTinterwordstretchfactor}{4}
\providecommand{\BIBentryALTinterwordspacing}{\spaceskip=\fontdimen2\font plus
\BIBentryALTinterwordstretchfactor\fontdimen3\font minus
  \fontdimen4\font\relax}
\providecommand{\BIBforeignlanguage}[2]{{%
\expandafter\ifx\csname l@#1\endcsname\relax
\typeout{** WARNING: IEEEtran.bst: No hyphenation pattern has been}%
\typeout{** loaded for the language `#1'. Using the pattern for}%
\typeout{** the default language instead.}%
\else
\language=\csname l@#1\endcsname
\fi
#2}}
\providecommand{\BIBdecl}{\relax}
\BIBdecl

\bibitem{knopp06}
R.~Knopp, ``Two-way radio networks with a star topology,'' in \emph{Proc. Int.
  Zurich Semin. Commun. (IZS)}, Zurich, Switzerland, Feb. 22--24 2006, pp.
  154--157.

\bibitem{rankovwittneben06}
B.~Rankov and A.~Wittneben, ``Achievable rate regions for the two-way relay
  channel,'' in \emph{Proc. IEEE Int. Symp. Inf. Theory (ISIT)}, Seattle, USA,
  July 9--14 2006, pp. 1668--1672.

\bibitem{rankovwittneben07}
------, ``Spectral efficient protocols for half-duplex fading relay channels,''
  \emph{IEEE J. Sel. Areas Commun.}, vol.~25, no.~2, pp. 379--389, Feb. 2007.

\bibitem{kattigollakota07}
S.~Katti, S.~Gollakota, and D.~Katabi, ``Embracing wireless interference:
  Analog network coding,'' in \emph{Proc. 2007 ACM SIGCOMM Conf.}, Kyoto,
  Japan, Aug. 27--31 2007, pp. 397--408.

\bibitem{gunduztuncel08}
D.~G{\"u}nd{\"u}z, E.~Tuncel, and J.~Nayak, ``Rate regions for the separated
  two-way relay channel,'' in \emph{Proc. 46th Allerton Conf. Commun. Control
  Comput. (Allerton Conf.)}, Monticello, USA, Sept. 23--26 2008, pp.
  1333--1340.

\bibitem{schnurrstanczak08}
C.~Schnurr, S.~Stanczak, and T.~J. Oechtering, ``Achievable rates for the
  restricted half-duplex two-way relay channel under a
  partial-decode-and-forward protocol,'' in \emph{Proc. IEEE Inf. Theory
  Workshop (ITW)}, Porto, Portugal, May 5--9 2008, pp. 134--138.

\bibitem{gunduzyener09}
D.~G{\"u}nd{\"u}z, A.~Yener, A.~Goldsmith, and H.~V. Poor, ``The multi-way
  relay channel,'' in \emph{Proc. IEEE Int. Symp. Inf. Theory (ISIT)}, Seoul,
  Korea, June 28--July 3 2009, pp. 339--343.

\bibitem{ongjohnsonkellett10cl}
L.~Ong, S.~J. Johnson, and C.~M. Kellett, ``An optimal coding strategy for the
  binary multi-way relay channel,'' \emph{IEEE Commun. Lett.}, vol.~14, no.~4,
  pp. 330--332, Apr. 2010.

\bibitem{ongmjohnsonit11}
------, ``The capacity region of multiway relay channels over finite fields
  with full data exchange,'' \emph{IEEE Trans. Inf. Theory: Special Issue on
  Interference Networks}, vol.~57, no.~5, pp. 3016--3031, May 2011.

\bibitem{wynerwolf02}
A.~D. Wyner, J.~K. Wolf, and F.~M.~J. Willems, ``Communicating via a processing
  broadcast satellite,'' \emph{IEEE Trans. Inf. Theory}, vol.~48, no.~6, pp.
  1243--1249, June 2002.

\bibitem{slepianwolf73}
D.~Slepian and J.~K. Wolf, ``Noiseless coding of correlated information
  sources,'' \emph{IEEE Trans. Inf. Theory}, vol. IT-19, no.~4, pp. 471--480,
  July 1973.

\bibitem{timoonglechner11}
R.~Timo, L.~Ong, and G.~Lechner, ``The two-way relay network with arbitrarily
  correlated sources and an orthogonal {MAC},'' in \emph{Proc. Data Compression
  Conf. (DCC)}, Snowbird, USA, Mar. 29--31 2011, pp. 253--262.

\bibitem{dueck81}
G.~Dueck, ``A note on the multiple access channel with correlated sources,''
  \emph{IEEE Trans. Inf. Theory}, vol.~27, no.~2, pp. 232--235, Mar. 1981.

\bibitem{coverthomas06}
T.~M. Cover and J.~A. Thomas, \emph{Elements of Information Theory},
  2nd~ed.\hskip 1em plus 0.5em minus 0.4em\relax Wiley-Interscience, 2006.

\end{thebibliography}

\end{document}